\documentclass[sigconf]{acmart}

\usepackage{multirow}
\usepackage{multicol}
\usepackage{siunitx}
\usepackage{comment}
\usepackage{balance}  
\usepackage{graphicx}
\usepackage{float}
\usepackage{subcaption}
\usepackage{xcolor}
\usepackage{makecell}
\usepackage{cancel}
\usepackage{enumitem}

\usepackage{tikz}
\usepackage{lipsum}
\usepackage{nicefrac}
\usepackage[ruled,vlined,linesnumbered]{algorithm2e}

\usepackage{soul}

\newcommand{\dist}{d}
\newcommand{\vcut}{V_{cut}}

\newcommand{\vborder}{V_{border}}

\newcommand*{\true}{\text{\normalfont True}}
\newcommand*{\false}{\text{\normalfont False}}

\tikzstyle{vertex}	= [circle,fill=black!20,minimum size=15pt,inner sep=0pt]
\tikzstyle{svertex}	= [circle,fill=black!20,minimum size=12pt,inner sep=0pt]
\tikzstyle{edge} 	= [draw,thick]
\tikzstyle{poss} 	= [font=\small]
\tikzset{top/.style={baseline=(current bounding box.north)}}
\tikzset{mid/.style={baseline=(current bounding box.center)}}

\usepackage{xcolor}

\newcommand{\q}[1]{\textcolor{brown}{#1}}
\newcommand{\f}[1]{\textcolor{blue}{#1}}

\newenvironment{updated}{\color{black}}{}

\usepackage{amsmath,amsthm}

\newtheorem{theorem}{Theorem}[section]
\newtheorem{lemma}[theorem]{Lemma}
\newtheorem{example}[theorem]{Example}

\newtheorem{definition}[theorem]{Definition}

\AtBeginDocument{%
  }

\setcopyright{acmcopyright}
\copyrightyear{2018}
\acmYear{2018}
\acmDOI{XXXXXXX.XXXXXXX}

\acmConference[Conference acronym 'XX]{Make sure to enter the correct
  conference title from your rights confirmation emai}{June 03--05,
  2018}{Woodstock, NY}
\acmPrice{15.00}
\acmISBN{978-1-4503-XXXX-X/18/06}

\begin{document}

\title{Hierarchical Cut Labelling \\-- Scaling Up Distance Queries on Road Networks}






\author{Muhammad Farhan$^{1}$, Henning Koehler$^{2}$, Robert Ohms$^{1}$, Qing Wang$^{1}$}
\affiliation{
\institution{$^{1}$Australian National University, Canberra, Australia, $^{2}$Massey University, Palmerston North, New Zealand}
\city{}
\country{}}
\email{{muhammad.farhan, qing.wang}@anu.edu.au; h.koehler@massey.ac.nz; robert.e.ohms@gmail.com}

\begin{abstract}
Answering the shortest-path distance between two arbitrary locations is a fundamental problem in road networks. Labelling-based solutions are the current state-of-the-arts to render fast response time, which can generally be categorised into hub-based labellings, highway-based labellings, and tree decomposition labellings. Hub-based and highway-based labellings exploit hierarchical structures of road networks with the aim to reduce labelling size for improving query efficiency. However, these solutions still result in large search spaces on distance labels at query time, particularly when road networks are large. Tree decomposition labellings leverage a hierarchy of vertices to reduce search spaces over distance labels at query time, but such a hierarchy is generated using tree decomposition techniques, which may yield very large labelling sizes and slow querying. In this paper, we propose a novel solution \emph{hierarchical cut 2-hop labelling (HC2L)} to address the drawbacks of the existing works. Our solution combines the benefits of hierarchical structures from both perspectives - reduce the size of a distance labelling at preprocessing time and further reduce the search space on a distance labelling at query time. At its core, we propose a new hierarchy, \emph{balanced tree hierarchy}, which enables a fast, efficient data structure to reduce the size of distance labelling and to select a very small subset of labels to compute the shortest-path distance at query time. To speed up the construction process of HC2L, we
further propose a parallel variant of our method, namely HC2L$^p$. We have evaluated our solution on 10 large real-world road networks through extensive experiments. The results show that our method is 1.5-4 times faster in terms of query processing while being comparable in terms of labelling construction time and achieving up to 60\% smaller labelling size compared to the state-of-the-art approaches.
\end{abstract}


\maketitle

\section{Introduction}
\begin{updated}
A \emph{distance query} in road networks is to find the length of a shortest path between any two given locations. This is a fundamental building block with numerous real-world applications, such as GPS navigation \cite{goldberg2005computing}, route planning \cite{fan2010improvement}, traffic monitoring \cite{kriegel2007proximity}, and point of interest (POI) recommendation \cite{yawalkar2019route}. Nowadays, computational resources such as memory and storage become readily available, e.g., road networks such as USA and EUR with tens of millions of vertices can be easily accommodated by a single server. However, many of these applications have low latency requirements, arising from the need to compute thousands to millions of distance queries per second as part of a more complex problem, which itself needs to be solved frequently, e.g. matching taxi drivers to passengers, optimizing delivery routes with multiple pick up and drop off points that can change dynamically, or providing recommendation on k-nearest POIs to their customers. 

For example, ride-hailing companies such as Uber are often required to compute millions of shortest-path distances between cars and customers each second \cite{zhang2021efficient,huang2021learning}
(e.g., between the locations of 1k cars and 10k customers) in order to process customers requests, e.g., finding nearest cars to customers. In such cases, even if answering one distance query only takes 1 microsecond, processing 10 million distance queries however would still amount to 10 seconds. Thus, it is very important to improve the performance of computing shortest-path distances for these applications to ensure good service quality to their customers. 



%
\end{updated}





\vspace{0.1cm}
\noindent\textbf{Related Work.~}
We briefly review the existing work for answering shortest-path distance queries in road networks.

\vspace{0.1cm}
\noindent\emph{\underline{Search-based approaches.}} A traditional approach to answering a distance query is to use the Dijkstra's algorithm \cite{tarjan1983data} which can compute the length of a shortest-path from a source vertex to a destination vertex in $O(|E| + |V|log|V|)$ time. 
This is however very slow in reality. Particularly, for pairs of vertices that are far apart from each other, the search space is large. To improve search efficiency, a bidirectional scheme can be used to run two Dijkstra's searches: one from the source vertex and the other from the destination vertex \cite{Pohl1969BidirectionalAH}. However, road networks are structurally characterised by high diameters and low node degrees, making search-based approaches such as Dijkstra's algorithm highly inefficient. In general, search-based approaches fail to achieve desired response time performance required by many real-world applications that operate on increasingly large road networks. 

To accelerate search by directing the search space towards useful vertices, rather than conducting unrestrained searches, a number of search-based approaches leverage indices. For example, ALT \cite{goldberg2005computing} pre-computes a partial distance index to accelerate the A* search. 
HiTi \cite{jung2002efficient} constructs a hierarchical structure using graph decomposition techniques to accelerate query performance. Highway Hierarchies \cite{10.1007/11561071_51} and Contraction Hierarchies (CH) \cite{geisberger2008contraction}  have shown to prune the search space massively for efficient querying. Moreover, algorithms such as Transit Node Routing (TNR) \cite{arz2013transit}, TRANSIT \cite{bast2006transit} and Arterial Hierarchy (AH) \cite{10.1145/2463676.2465277} integrate coordinate information and divide a road network into multiple hierarchical grids.
Transit Node Routing and TRANSIT then precompute an index which stores distances between vertices with respect to grids. Arterial Hierarchy is an improved version of the CH algorithm. Despite considerable progress, these algorithms still require exploring a large search space for distance queries when vertices are far apart from each other in a road network.

\vspace{0.1cm}
\noindent\emph{\underline{Labelling-based approaches.}}
To address the limitations of search-based approaches, 
 labelling-based approaches for road networks have been studied with great success \cite{abraham2011hub,abraham2012hierarchical,akiba2014fast,ouyang2018hierarchy,akiba2013fast,chen2021p2h,jin2012highway,zhang2022relative}. Instead of performing searches on a road network at query time, these approaches precompute distance labels that fully capture distance information between pairs of vertices, and then perform searches on distance labels at query time to compute distances. Labelling-based approaches can answer distance queries significantly faster than search-based approaches, at the cost of requiring additional space for storing precomputed labels. As such, it is still an open problem how to design algorithms which can give both fast query times with small memory costs. 

The current state-of-the-art labelling-based approaches exploit hierarchical structures of road networks. Most of these approaches satisfy the 2-hop cover property \cite{cohen2003reachability} which requires at least one common vertex in the distance labels $L(s)$ and $L(t)$ to be on a shortest-path between any query pair $(s, t)$, generally falling into three categories: hub-based labellings, highway-based labellings, and tree decomposition labellings. Approaches for hub-based labellings \cite{abraham2011hub,abraham2012hierarchical} generate distance labels for vertices that contain the distances to hub vertices following a vertex ordering computed by conducting CH searches on a road network. This results in a hierarchical structure among distance labels which is closely relating to the labelling size. Approaches for highway-based labellings \cite{akiba2014fast} decompose a road network into disjoint shortest-paths and then construct distance labels for vertices that contain the distances to a subset of decomposed shortest-paths. Similar to the importance of a vertex ordering to hub-based labellings, the order of shortest-paths is important to highway-based labellings for small labelling sizes.
Approaches for tree decomposition labellings \cite{ouyang2018hierarchy,chen2021p2h} first find a hierarchy over vertices in a road network using tree decomposition techniques \cite{bodlaender2006treewidth}. 
Then a distance query $(s, t)$ is answered by searching such a tree-decomposition hierarchy to identify a subset of common vertices in the distance labels $L(s)$ and $L(t)$, and compute the distance between $s$ and $t$.



\vspace{0.1cm}
\noindent\textbf{Present Work.~}From the above analysis, we observe that labelling-based approaches leverage hierarchies over vertices in a road network for distance queries in two ways. One way is to reduce the size of distance labellings. In hub-based labellings and highway-based labellings, a distance query $(s, t)$ needs to search through entire labels $L(s)$ and $L(t)$ when computing a shortest-path distance. Therefore, they aim to find a ``good ordering'' of vertices which can produce distance labels of small sizes, and reduce the search space on $L(s)$ and $L(t)$ in the process. The other way is to assign a hierarchy over vertices, and while querying, use such a hierarchy to reduce the search space to a small subset of entries in a distance labelling. In tree-decomposition labellings, such a hierarchy is obtained by applying existing tree decomposition techniques~\cite{bodlaender2006treewidth}. 


Motivated by the above observations, in this work, we aim to design a solution that combines the benefits of using vertex hierarchies from both perspectives - leveraging a vertex hierarchy to reduce the size of a distance labelling and further reduce the search space on such a distance labelling at query time. This leads to the following questions: 
\begin{itemize}
[leftmargin=*]
\item What are the desirable characteristics of such a hierarchy?
\item How can a hierarchy be designed to reduce the sizes of distance labellings and find ``good hops'' to accelerate search on distance labels simultaneously? 
\item What are the hierarchical properties of distance labellings under this framework?
\end{itemize}

\vspace{0.1cm}
\noindent\textbf{Contributions.~}In this paper, we propose a new labelling-based method which addresses the above questions. Our contributions are summarised as follows:

\begin{itemize}[leftmargin=*]
\item Upon analysing existing approaches that exploit hierarchical structures of road networks for efficient distance querying, we introduce a new hierarchy called \emph{balanced tree hierarchy} which enables us to reduce the size of a distance labelling significantly. Further, our balanced tree hierarchy allows us to answer a distance query by processing a small subset of labels that are necessary to compute the distance. Compared to the state-of-the-art approaches, our method achieves much faster query times.

\item We develop an efficient algorithm to construct our proposed \emph{balanced tree hierarchy}. Our algorithm recursively partitions a road network while preserving two conditions: balanced partitions and small cuts. Specifically, it first finds a balanced and small vertex cut and then creates a minimum number of shortcuts to preserve the distance between border vertices. This results in a tree structured hierarchy among vertices, supported by an efficient data structure for finding hub vertices of any two vertices in a road network.

\item We propose a distance labelling called \emph{hierarchical cut 2-hop labelling} (HC2L) which is hierarchical in terms of a vertex quasi-order defined by a balanced tree hierarchy. The labels of any two given vertices must contain a common hub vertex that can be found via their lowest common ancestor in the balanced tree hierarchy. We analyse upper and lower bounds on labelling size of HC2L and devise a novel pruning strategy, called \emph{tail pruning}, that allows our proposed algorithm to produce a HC2L with smaller labelling sizes without compromising query efficiency.

\item We conduct extensive experiments to evaluate the performance of the proposed method on 10 real-world large road networks including the whole road network of the USA. The experimental results demonstrate that our method is 1.5-4 times faster than the state-of-the-art approaches in terms of query processing while consuming comparable labelling construction time and up to 60\% smaller labelling size.
    
\end{itemize}
\section{Preliminaries}
Let $G=(V,E)$ be a road network where $V$ is a set of vertices, and $E$ is a set of edges. Each edge $(u, v) \in E$ is associated with a positive weight $\omega(u,v) \in \mathbb{R}$. Given a set of vertices $S\subseteq V$, $G[S]=(S,\{(v,v')\in E|v,v'\in S\})$ is an induced subgraph of $G$ formed from $S$. A path is a sequence of vertices $p = (v_1, v_2, \dots, v_k)$ where $(v_{i}, v_{i+1}) \in E$ for each $1 \leq i < k$. The weight of a path $p$ is defined as $\omega(p) = \sum_{i=1}^{k - 1} w(v_i, v_{i+1})$. For two arbitrary vertices $s$ and $t$, a shortest path $p$ between $s$ and $t$ is a path starting at $s$ and ending at $t$ such that $\omega(p)$ is minimised. The distance between $s$ and $t$ in $G$, denoted as $d_G(s,t)$, is the weight of any shortest path between $s$ and $t$. We use $N(v)$ to denote the set of direct neighbors of a vertex $v \in V$, i.e. $N(v) = \{u \in V \mid (u, v) \in E \}$, and $V(G)$ and $E(G)$ to refer to the set of vertices and edges in $G$, respectively. Each vertex $v \in V$ is associated with a {label} $L(v)$. The set of labels $L = \{L(v) \ | \ v \in V\}$ is called a \emph{distance labelling} over $G$. A \emph{vertex cut} $\vcut \subseteq V$ on $G$ is a subset of vertices whose removal from $G$ splits $G$ into multiple connected components. Vertices in $\vcut$ are called \emph{cut vertices}.

The distance query problem on road networks is defined below.
\begin{definition}[Problem Definition]
Given a road network $G=(V, E)$, the \emph{distance query problem} on $G$ is to compute the distance $d_G(s, t)$ between any two arbitrary vertices $s, t \in V$. 
\end{definition}
In this work, we study labelling-based techniques to efficiently answer distance queries on road networks.

\section{Existing Solutions}\label{sec:existing-solutions}
The currently fastest known solutions for the shortest-path distance problem on a road network compute a {distance labelling} $L$. Such a distance labelling usually satisfies a \emph{2-hop cover property} \cite{cohen2003reachability}, requiring that the labels of any two vertices  must contain at least one common vertex on their shortest-paths, hence called \emph{2-hop labelling}. In the literature, there are three popular labelling techniques that exploit hierarchical structures of road networks for computing 2-hop labellings: 
1) hub-based labellings \cite{abraham2011hub,abraham2012hierarchical}, 2) highway-based labellings \cite{akiba2014fast}, and 3) tree-decomposition labellings \cite{ouyang2018hierarchy,chen2021p2h}. We now discuss them in detail.

\begin{figure}[t]
\includegraphics[scale=0.7]{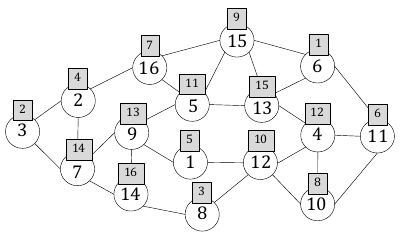}
\begin{center}
(a)    
\end{center}
\includegraphics[scale=0.7]{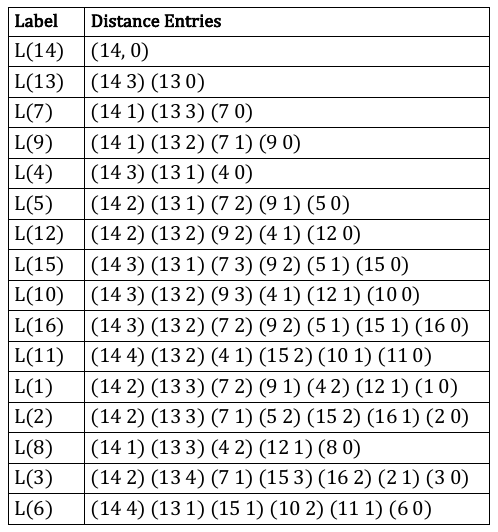}
\begin{center}
(b)
\end{center}
\caption{(a) An example road network, (b) Hub labelling.}
\label{fig:hl_labelling}

\end{figure}


\subsection{Hub-Based Labellings}\label{subset:hub-based labelling}
Abraham et al.~\cite{abraham2011hub} proposed the first hub-based labelling algorithm, called \emph{{H}ub-based {L}abeling} (HL), to construct labels by storing the distances from each vertex to hub vertices (i.e., vertices on shortest paths). The label of each vertex $v$ is thus a set of {distance entries} $\{(u_1, \delta_{vu_1}), \dots, (u_k, \delta_{vu_k})\}$ where $\{u_1,\dots,u_k\}\subseteq V$ and $\delta_{vu_i}=d_G(v,u_i)$. 
Their work was motivated by the observation that vertices visited by searches of hierarchical and reach-based algorithms \cite{sanders2005highway,gutman2004reach} form a 2-hop labelling. It turns out that HL produces 2-hop labellings that are much smaller than the worst-case bounds~\cite{gavoille2004distance}.
Hub-based labellings are not necessarily hierarchical.
Later, Abraham et al.~\cite{abraham2012hierarchical} studied hierarchical hub labellings with respect to a hierarchy defined by the relationship ``vertex $v$ is in the label of vertex $u$". For any hierarchical hub labelling $L$, there exists a \emph{canonical labelling}, which is the smallest hierarchical hub labelling with respect to all total vertex orderings that are consistent with the hierarchy of $L$.

Generally, given any two vertices $s$ and $t$, hub-based labelling methods compute the distance $d_G(s,t)$ by searching the labels $L(s)$ and $L(t)$ as follows:
\begin{equation}\label{equ:2-hop}    
d_G(s,t)=min\{\delta_{su}+\delta_{tu}:(u,\delta_{su})\in L(s),(u,\delta_{tu})\in L(t)\}.
\end{equation}


\begin{example}
Figure~\ref{fig:hl_labelling}(b) shows the canonical labelling for the road network shown in Figure \ref{fig:hl_labelling}(a) which respects a total vertex ordering
$14{>}13{>}7{>}9{>}4{>}5{>}12{>}15{>}10{>}16{>}11{>}1{>}2{>}8{>}3{>}6$.
The order of each vertex is shown in a box at the top in Figure \ref{fig:hl_labelling}(a). 
We can see from Table~\ref{fig:hl_labelling} that the labels of any two vertices contain the distance to the highest ranked vertex on their shortest-paths. Consider two vertices $3$ and $11$, they have one shortest-path $(3, 2, 16, 15, 6, 11)$, among which vertex $15$ is of the highest order and thus stored in both labels $L(3)$ and $L(5)$.
\end{example}


\subsection{Highway-Based Labellings}
Akiba et al.~\cite{akiba2014fast} proposed a highway-based labelling algorithm, namely  \emph{Pruned Highway Labelling} (PHL), which generalizes hub-based labellings by explicitly exploiting \emph{highway} structure in road networks. Highways are considered as shortest-paths which are central, i.e., passed through by many other shortest-paths.
Their algorithm first decomposes a road network into a set of disjoint shortest-paths $P$ and then computes a label $L(v)$ for each vertex $v$ which stores the distance from $v$ to the shortest-paths in $P$. Each label $L(v)$ is a set of triples $\{(p_i, \delta_{u_iu_j},\delta_{vu_j})\}_{p_i\in P}$ where $p_i\in P$ referring to a shortest-path, $\delta_{u_iu_j}=d_G(u_i,u_j)$ referring to the distance from the starting vertex $u_i$ to another vertex $u_j$ on the shortest path $p_i$, and $\delta_{vu_j}=d_G(v,u_j)$. 
To reduce labelling size, a pruned Dijkstra’s
search is conduced from each shortest-path to construct labels, similar to pruned landmark labelling~\cite{akiba2013fast}.

Given any two vertices $s$ and $t$, the distance $d_G(s,t)$ is computed by searching the labels $L(s)$ and $L(t)$ such that
\begin{align}
d_G(s,t)=&min\{\delta_{su_j}+\delta_{tu_j'}+|\delta_{u_iu_j}-\delta_{u_iu_j'}|: \\\nonumber
& (p_i,\delta_{u_iu_j},\delta_{su_j})\in L(s), (p_i, \delta_{u_iu_j'},\delta_{tu_j'})\in L(t), p_i\in P\}.
\end{align}

\begin{figure}[h]
\includegraphics[scale=0.7]{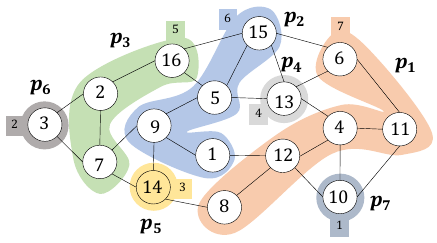}
\begin{center}
(a)
\end{center}
\includegraphics[scale=0.7]{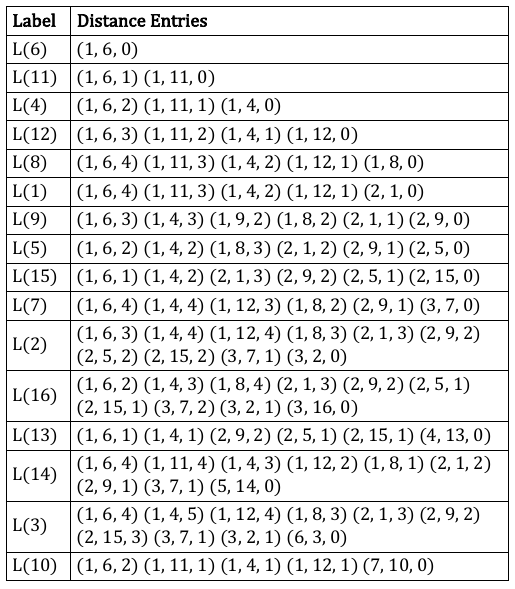}
\begin{center}
(b)
\end{center}
\caption{(a) Highway decomposition into shortest-paths, (b) Pruned highway labelling.}
\label{fig:phl_labelling}\vspace{-0.1cm}
\end{figure}
\begin{example}
Figure \ref{fig:phl_labelling}(a) shows the highway decomposition of the road network into shortest-paths $P=\langle p_1, p_2, p_3, p_4, p_5, p_6, p_7\rangle$ which are highlighted in different colors. $P$ defines a partial vertex ordering as $p_1>p_2>p_3>p_4>p_5>p_6>p_7$. The order of vertices in each path is shown in a box beside each vertex. Table~\ref{fig:phl_labelling}(b) shows the labelling constructed by PHL using the order implied by $P$. From Figures \ref{fig:phl_labelling}(a) and Table~\ref{fig:phl_labelling}(b), we can see that the label of each vertex only stores the distances to vertices in the paths with the same or higher orders. Consider two vertices $4$ and $8$ in $p_1$, the labels $L(4)$ and $L(8)$ only store the triples from vertices in $p_1$. 
\end{example}





\subsection{Tree-Decomposition Labellings} 
Recently, Ouyang et al.~\cite{ouyang2018hierarchy} proposed a tree-decomposition labelling method, called \emph{Hierarchical 2-Hop Index} (H2H), to exploit \emph{tree decomposition} structures in road networks. The general idea is to not only construct a 2-hop labelling but also generate a hierarchy among all vertices in a road network using tree decomposition techniques~\cite{DBLP:journals/actaC/Bodlaender93}. Based on such a hierarchy, a subset of vertices in labels are visited when answering distance queries. In a later work \cite{chen2021p2h}, pruning techniques were proposed to further reduce the number of vertices in labels being visited at query time.

One important property leveraged by tree-decomposition labelling methods is that, for any graph $G$ and any tree decomposition $T_G$, each internal node of $T_G$ represents a vertex cut on $G$. This allows to leverage the {lowest common ancestor} $LCA(s, t)$ of any two vertices $s$ and $t$ in $T_G$ to find a vertex cut that separates $s$ and $t$. Accordingly, a 2-hop labelling is constructed in H2H such that the label $L(v)$ of each vertex $v\in V(G)$ consists of two arrays: a \emph{distance array} $(\delta_{vw_1},\dots, \delta_{vw_k})$ where $\delta_{vw_i}=d_G(v,w_i)$ and $\{w_1,\dots, w_k\}$ is the set of vertices that are ancestors of $v$ in $T_G$, and a \emph{position array} $(p_1,\dots, p_k)$ that stores positions to $\{w_1,\dots, w_k\}$ in $T_G$. Let $L(v).dist$ and $L(v).post$ denote the distance and position arrays in $L(v)$, respectively.  The distance between any two given vertices $s,t\in V(G)$ is computed as
\begin{align}\label{lab:h2h_query}
d_G(s,t)=&min\{\delta_{su_i}+\delta_{tu_i}: \delta_{su_i}=L(s).dist(i), \\\nonumber
&\delta_{tu_i}=L(t).dist(i), i\in L(w).post, w= LCA(s, t)\}.
\end{align}

\begin{figure}[h]
\centering
\includegraphics[width=0.48\textwidth]{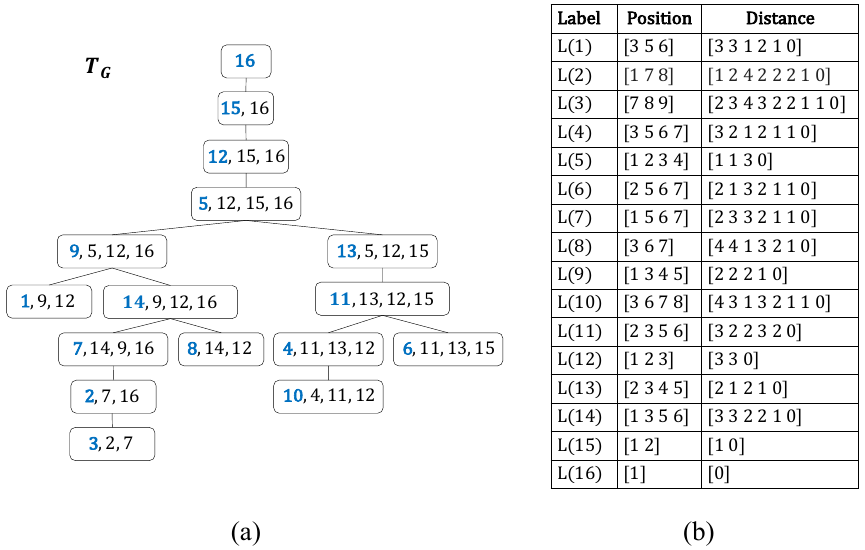}
\caption{Tree decomposition labelling: (a) a tree decomposition $T_G$ of the example road network $G$ shown in Figure \ref{fig:hl_labelling}(a), (b) H2H-Index.}
\label{fig:h2h_hierarchy}
\end{figure}

\begin{example}
Figure \ref{fig:h2h_hierarchy}(a) illustrates a tree decomposition $T_G$ of the road network shown in Figure \ref{fig:hl_labelling}(a) while Figure~\ref{fig:h2h_hierarchy}(b) shows the H2H index. The label of each vertex stores a position array and a distance array. Consider two vertices $7$ and $13$, the label $L(7)$ stores a position array $[1, 5, 7]$ which represents the positions (depths) of the ancestors $\{14, 9, 16\}$ inside the node associated with the vertex $7$ in $T_G$ shown in Figure \ref{fig:h2h_hierarchy}(a), and a distance array $[2, 3, 3, 2, 1, 1, 0]$ which represents the distances from $7$ to all the ancestors of vertex $7$ in $T_G$. For a distance query between vertices $7$ and $13$, $LCA(7, 13)=5$ is first obtained, and then using the distances in $L(7)$ and $L(13)$ at positions $[1, 2, 3]\in L(5)$, $d_G(7, 13) = 3$ is obtained according to Equation~\ref{lab:h2h_query}.
\end{example}

\subsection{Discussion}

Existing hub-based and highway-based labelling solutions search over all the distance entries in labels $L(s)$ and $L(t)$ for a query pair $(s, t)$. Thus, in order to accelerate querying, they exploit a vertex ordering to reduce size of their 2-hop labellings so that they can process smaller labels at query time. It is known that a vertex ordering is crucial for label construction. However, finding a 2-hop labelling of minimum size is NP-hard and finding a vertex ordering that minimizes labelling size is difficult~\cite{babenko2015complexity,akiba2013fast,cohen2003reachability}. 
Thus, hub-based and highway-based labelling solutions still suffer from scalability issues resulting in slow querying when a road network is large.


Existing tree-decomposition labelling solutions assume that a tree decomposition of a road network is given \cite{ouyang2018hierarchy,chen2021p2h}. However, obtaining a tree decomposition with the minimal tree width is a very difficult problem. Even for determining whether a graph $G$ has a tree width of at most a given value, it is known to be NP-complete \cite{arnborg1987complexity}. Sub-optimal algorithms exist \cite{bodlaender2006treewidth}, which can compute tree decomposition with a time complexity of $O(|V| \cdot (w^2$ + log$(|V|)))$ but may yield very large tree width $w$ and height $h$. Larger tree width and height take longer in constructing 2-hop labellings and produce larger labelling sizes ($|V| \cdot h$), thus resulting in slow querying. In addition to this, existing tree-decomposition labelling solutions use the Range Minimum Query (RMQ) based algorithms \cite{10.1007/10719839_9} to compute LCA for any two given vertices at query time. Although these algorithms can compute LCA in a tree in time $O(1)$, there are some hidden constant factors and additional computational requirements to achieve this. For instance, they need to precompute a data structure in order to store the information about LCA of all pairs of vertices in a tree. This data structure incurs significant computational overheads as shown in our experiments. 

The following example demonstrates the drawbacks of these existing solutions.

\begin{example}
Consider a distance query $(3, 10)$, HL in Figure \ref{fig:hl_labelling}(b) processes 7 distance entries in $L(3)$ and 6 distance entries in $L(10)$ to compute the distance $d_G(3,10)$. PHL in Figure \ref{fig:phl_labelling}(b) processes 10 distance entries in $L(3)$ and 5 distance entries in $L(10)$ to compute the distance $d_G(3,10)$. However, it is unnecessary to process all of these distance entries because only a small subset of distance entries in labels $L
(3)$ and $L(10)$ is sufficient for computing $d_G(3,10)$. For instance, only distance entries $(14, 2) \in L(3)$ and $(14, 3) \in L(10)$ in the hub labelling and distance entries $(1, 12, 4)\in L(3)$ and $(1, 12, 1)\in L(10)$ in the highway labelling are needed when computing the distance. 

In Table \ref{fig:h2h_hierarchy}(b), H2H stores distances to all ancestors of each vertex along with positions of the ancestors appearing in their associated nodes in $T_G$. This results in a large labelling size because vertices may appear in multiple nodes in $T_G$. Further, to compute LCA in constant time, this would require an extra space overhead. For instance, this extra space overhead is 4.64 GB for the whole USA (USA) dataset and 3.69 GB for Western Europe (EUR) dataset as shown in Table \ref{table:lca_hopsize}. 
\end{example}

\section{Our Solution}
In this section, we present a novel labelling-based solution, referred to as \underline{H}ierarchical \underline{C}ut \underline{2}-Hop \underline{L}abelling (HC2L) framework, for shortest-path distance queries on road networks. Let $G$ be a road network. Our solution is to build an efficient \emph{distance scheme} $S=(H_G, L_G, Q)$ over $G$, where $H_G$ is a vertex hierarchy on $G$, $L_G$ is a distance labelling on $G$ satisfying the 2-hop cover property, and $Q$ is a query function which, given any two vertices $s,t\in V(G)$, computes $d_G(s,t)$ based on their labels in $L$ and their hierarchical positions in $H_G$. 
\subsection{Hierarchy Construction}



To design an efficient distance scheme, the choice of a hierarchy is crucial. We observe that binary tree structure may serve as an efficient indexing scheme for road networks. This is because vertices can be indexed in a way that the LCA of any two indexed vertices contains at least one hub vertex on their shortest paths, enabling us to efficiently find them. It is also desirable to keep the height of such a binary tree small, thereby being as balanced as possible. Below, we define \emph{balanced tree hierarchy}.




\begin{definition}[Balanced Tree Hierarchy]\label{def:td}
Let $\beta$ be a balancing-parameter with $0<\beta\leq 0.5$. A \emph{balanced tree hierarchy} is a binary tree $H_G=(\mathcal{N}, \mathcal{E}, \ell)$, where $\mathcal{N}$ is a set of tree nodes, $\mathcal{E}$ is a set of tree edges, and $\ell: V(G)\rightarrow \mathcal{N}$ is a total surjective function. Further, $H_G$ must satisfy the following conditions:
\begin{enumerate}
\item For any internal tree node $N_i\in\mathcal{N}$, its subtrees are balanced:
\begin{equation*}
|\textsc{Left}(N_i)|, |\textsc{Right}(N_i)| \leq (1-\beta) \cdot |\textsc{Subtree}(N_i)|
\end{equation*}
where $\textsc{Right}(\cdot)$ and $\textsc{Left}(\cdot)$ refer to the vertices mapped into the right and left subtrees of a tree node, respectively,
and $\textsc{Subtree}(\cdot)$ to vertices mapped into the whole subtree.
\item For any two vertices $s,t\in V(G)$, the lowest common ancestor of their tree nodes $LCA (s, t)$ contains at least one vertex on a shortest-path between $s$ and $t$.
\end{enumerate}
\end{definition}


Thus, for each tree node at the $k$-th level of $H_G$, its subtree has at most $V\cdot(1-\beta)^k$ vertices. This leads to the following lemma.

\begin{lemma}
Let $\alpha=1/(1-\beta)$. The height of $H_G$ is bounded by $\log_\alpha(n)$, where $n=|V(G)|$.
\end{lemma}



In the following we introduce our algorithm for constructing a balanced tree hierarchy $H_G$. At its core, the algorithm recursively computes a balanced cut to partition a road network into smaller components, collectively named \emph{hierarchical balanced cuts}.
\begin{updated}
Note that this is closely related to the minimum balanced vertex separator problem, which is known to be NP-hard \cite{DBLP:conf/stoc/FeigeM06}.
\end{updated}

\subsubsection{Hierarchical Balanced Cuts}
Let $G$ be a road network. The algorithm recursively bisects a graph $G$ in two steps: (1) \emph{Balanced partitioning}: partition an input graph $G'\subseteq G$ into two initial partitions connected via another partition referred to as a \emph{cut region}; (2) \emph{Minimal vertex cuts}: find a minimal vertex cut within the cut region. Each iteration of the algorithm splits $G'$ into two smaller but balanced components that contain the initial partitions, while a minimal vertex cut within the cut region ensures that vertices in such a cut are central in $G'$, i.e., passed through by many shortest-paths between vertices. This is illustrated in Figure~\ref{fig:partition_tree_labelling}(a). In the following, we discuss these steps in detail.

\vspace{0.1cm}
\noindent\textbf{Balanced Partitioning.~} Algorithm~\ref{algo:rough-partition} provides the pseudo-code for this approach. We start with two vertices $v_A$ and $v_B$ as far apart as possible (Lines 11-12), and assign to each vertex $v$ a \emph{partition weight} (Line 13):
\[
pw(v) := \dist_{G'}(v_A,v) - \dist_{G'}(v_B,v).
\]
Partition weights split vertices into equivalence classes. 
\begin{definition}[Equivalence Class]
Given $v\in V(G')$, the equivalence class of $v$ in $G'$ is $\{v' \in V(G') \mid pw(v) = pw(v')\}$.
\end{definition}
Then, two initial partitions $P'_A$ and $P'_B$ are created by picking vertices with lowest and highest partition weights, respectively, until the desired balancing condition (e.g. $\beta=0.3$) is satisfied (Lines 14-15).
The remaining vertices form a region, called a \emph{cut region}.

However, there is a complication that needs to be addressed. 
Suppose that we have $v_A=2$ and $v_B=3$ in the partition $P_A$ shown in Figure \ref{fig:partition_tree_labelling}(a), the vertex $7$ constitutes a bottleneck that causes all vertices whose shortest-paths to $2$ and $3$ pass through it to have the same equivalence class because of same partition weight as $7$. As a result, nodes from this equivalence class are added to $P'_A$ and $P'_B$ arbitrarily, leading to large cuts. To address this issue, we detect and remove the bottleneck from the graph temporarily and repeat the process for finding a rough partition on the remaining graph (Lines 16-21). The bottleneck itself is then also added to the cut region (Line 22). As this bottleneck removal can lead to the graph being disconnected, we first compute the connected components of $G$ (Line 3). If the largest component $C_{\max}$ contains no more than $(1-\beta)\cdot|V|$ nodes, the empty cut is already balanced (Lines 9-10). Otherwise we find our initial partition within $C_{\max}$ (Lines 5-7).

\begin{updated}
It is worth noting that the choice of initial vertices only determines the initial partition and the real optimization happens during the vertex cut step. Picking distant vertices is mainly important to ensure that initial partitions are well separated, thus leaving enough room for a vertex cut to avoid dense clusters.
We ran additional experiments (not reported) where we made multiple random choices for the arbitrary node in line~\ref{ln:rp-arbitrary}, then computed a partition for each choice, and picked the one with the smallest cut. This resulted in only minor reductions to labelling size (less than 5\% for the graphs examined), and we concluded that further optimization of our initial vertex choices (beyond being far apart) would be unlikely to justify the resulting increase in construction time.
\end{updated}

\begin{algorithm}[t]
\caption{Balanced Partition}\label{algo:rough-partition}
\SetCommentSty{textit}
\SetKwFunction{FMain}{BalancedPartition}
\SetKwProg{Fn}{Function}{}{end}
\Fn{\FMain{$G$, $\beta$}}{
	\If{$G$ is disconnected} {
	    $CC\gets$ connected components of $G$\\
		$C_{\max}\gets$ largest connected component\\
		\If{$|C_{\max}| > (1-\beta)\cdot |V|$}{
			$(P'_A,C,P'_B)\gets\FMain(C_{\max},\beta)$ \\
			\Return $(P'_A,\; C\cup(V\setminus C_{\max}),\; P'_B)$ \label{ln:rp-return-1}
		}
		\Else {
            $C_{2nd}\gets$ 2nd largest connected component\\
            \Return $(C_{\max},\; V\setminus(C_{\max}\cup C_{2nd}),\; C_{2nd})$ \label{ln:rp-return-2}
		}
	}
    \tcp{find distant nodes using BFS or Dijkstra}
    $v_A\gets$ node furthest away from arbitrary node \label{ln:rp-arbitrary}\\
    $v_B\gets$ node furthest away from A\\
    \tcp{find balanced partitions}
    sort nodes by $pw(v)=d_G(v_A,v) - d_G(v_B,v)$\\
    $P'_A\gets$ first $\beta\cdot|V|$ nodes\\
    $P'_B\gets$ last $\beta\cdot|V|$ nodes\\
	\tcp{handle bottlenecks}
	$w_A\gets\max\{pw(v) \mid v\in P'_A\}$\\
	$w_B\gets\min\{pw(v) \mid v\in P'_B\}$\\
	\If{$w_A = w_B$}{ \label{ln:rp-bottleneck}
		$EQ\gets \{ v\in V \mid pw(v) = w_A\}$\\
		$\mathcal{B}\gets \arg\min_{v\in EQ} d_G(v_A,v)$\\
		$(P'_A,C,P'_B)\gets\FMain(G\setminus\mathcal{B},\beta)$\\
		\Return $(P'_A,\; C\cup\mathcal{B},\; P'_B)$ \label{ln:rp-return-3}
	}
	\tcp{ensure $P'_A,P'_B$ are connected}
	$P'_A\gets \{ v\in V \mid pw(v) \leq w_A \}$ \label{ln:rp-connected-a}\\
	$P'_B\gets \{ v\in V \mid pw(v) \geq w_B \}$ \label{ln:rp-connected-b}\\
    \Return $(P'_A,\; V\setminus(P'_A\cup P'_B),\; P'_B)$ \label{ln:rp-return-4}
}
\end{algorithm}

\vspace{0.1cm}
\noindent\textbf{Minimal Vertex Cuts.~} Algorithm~\ref{algo:partition} shows the pseudo-code for this approach. To find a minimal vertex cut within a cut region, we case this as a minimal ($S, T$)-vertex-cut problem by contracting vertices in the initial partitions $P'_A$ and $P'_B$ into single vertices $S$ and $T$, respectively (Lines 3-11).
Here $S$ is adjacent to a vertex $v$ in the cut region iff any vertex in $P'_A$ is, and similar for $T$. This allows us to reduce this problem to a maximum flow problem using the well-known graph transformation technique~\cite{Bondy1976}
and solve it using a variant of Dinitz's algorithm \cite{DBLP:conf/birthday/Dinitz06} (Line 12). 


The flow graph for the example graph in Figure \ref{fig:hl_labelling}(a), resulting from the transformation in \cite{Bondy1976}, is shown in Figure~\ref{fig:flow_graph}(b). \emph{Inner} edges (shown in red) connect the incoming and outgoing copies of a node, and have capacity one. \emph{Outer} edges connect copies of different nodes, and have infinite capacity. Since flow paths alternate between inner and outer edges, and each vertex is incident to only one inner edge, outer edges are effectively limited to capacity one as well. For unit capacity graphs, Dinitz's algorithm requires at most $O(\sqrt{|V|})$ phases \cite{DBLP:conf/birthday/Dinitz06}, and each phase runs in $O(|E|)$.
As each phase increases the flow by at least one, the number of phases is also bounded by $|V_{cut}|$. This results in a complexity bound of $O(|E|\cdot\min(\sqrt{|V|}, |V_{cut}|))$.

When extracting a minimal vertex cut from the maximal flow, we have two options: pick for each flow path the node closest to $S$ (amongst those that cannot be reached from $S$ in the residual graph), or the node closest to $T$ (amongst those that cannot reach $T$ in the residual graph). We evaluate both options and pick the more balanced one. For the flow graph in Figure~\ref{fig:flow_graph}(b), this means we pick $\{16,5,12\}$ over $\{15,13,12\}$. 
Once a vertex cut $V_{cut}$ has been found and removed from the graph, we assign connected components to either $P_A$ or $P_B$ while maximizing balance (Lines 14-15).


\begin{figure}[h]
\centering
\includegraphics[width=0.48\textwidth]{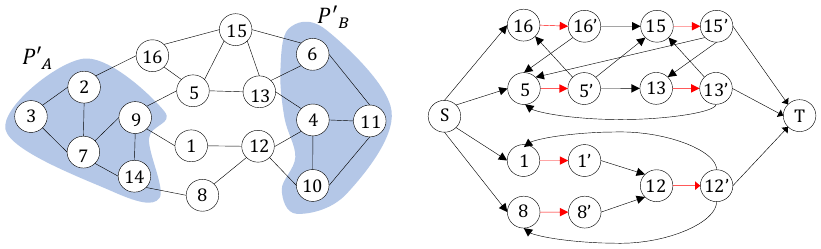}
\begin{center}
(a)\hspace{4cm}(b)
\end{center}
\caption{(a) Initial partitions, (b) Flow graph.}
\label{fig:flow_graph}
\end{figure}

However, if $S$ and $T$ end up with an edge between them, due to an edge between $a\in S$ and $b\in T$, there exists no vertex cut in the cut region.
This case occurs mostly at the lower levels of a tree hierarchy. To maintain balance guarantees, we move $a$ and $b$ to the cut region, and connect them to $S$ and $T$, respectively. This ensures that one of them becomes a cut vertex and the other remains in its original partition. 


\begin{lemma}\label{L:partition-complexity}
Algorithm~\ref{algo:rough-partition} runs in $O(|E|\cdot\log |V|\cdot k)$, where $k$ is the number of times the check in line~\ref{ln:rp-bottleneck} succeeds.
Algorithm \ref{algo:partition} runs in $O(|E|\cdot(\log |V|\cdot k + \min(|V_{cut}|,\sqrt{|V|})))$.
\end{lemma}
In practice we found that $k\leq 1$ and $\log|V| \leq |V_{cut}| \leq \sqrt{|V|}$, except for very small subgraphs, resulting in effective time complexities of $O(|E|\cdot\log|V|)$ and $O(|E|\cdot|V_{cut}|)$.


\begin{algorithm}[t]
\caption{Balanced Cut}\label{algo:partition}
\SetCommentSty{textit}
\SetKwFunction{FMain}{BalancedCut}
\SetKwFunction{FRough}{BalancedPartition}
\SetKwProg{Fn}{Function}{}{end}
\Fn{\FMain{$G$, $\beta$}}{
    $(P'_A,C,P'_B)\gets\FRough(G,\beta)$ \\
    \tcp{find vertices in cross-partition edges}
    $C_A\gets P'_A\cap N(P'_B)$\\
    $C_B\gets P'_B\cap N(P'_A)$\\
    \tcp{construct s-t flow graph}
    $G_{ST}\gets G[C\cup C_A\cup C_B]\cup \{s,t\}$\\
    $N_S\gets C_A\cup \big(C\cap N(P'_A\setminus C_A)\big)$\\
    $N_T\gets C_B\cup \big(C\cap N(P'_B\setminus C_B)\big)$\\
    \ForEach{$v\in N_S$}{
        add edge $(s,v)$ to $G_{ST}$\\
    }
    \ForEach{$v\in N_T$}{
        add edge $(t,v)$ to $G_{ST}$\\
    }
    \tcp{find minimum cut (Dinitz's algorithm)}
    $\vcut\gets$ minimum s-t vertex cut on $G_{ST}$\\
    \tcp{construct final partitions}
    $CC\gets$ connected components of $G\setminus\vcut$\\
    \ForEach{$cc\in CC$ in order of decreasing size}{
        add $cc$ to smaller of $P_A$, $P_B$\\
    }
    \Return $(P_A,\vcut,P_B)$
}
\end{algorithm}

\subsubsection{Distance Preservation}
There is one issue with applying Algorithm~\ref{algo:partition}. Although it can partition a graph $G$ into two balanced components $P_A$ and $P_B$ via a vertex cut $\vcut$, the induced subgraphs of $G$ by $P_A$ and $P_B$ are not necessarily distance-preserving.


\begin{definition}[Distance-Preserving Property]
Let $G[P]$ denote an induced subgraph of $G$ by the vertices in a partition  $P$. We say that $P$ is \emph{distance-preserving} iff the following condition is satisfied:
\begin{equation}\label{equ:distance}
    \forall x,y\in P. \dist_{G[P]}(x,y)=\dist_G(x,y).
\end{equation}
\end{definition}

The example below illustrates the distance-preserving property.

\begin{example}\label{exa:shortcut}
Consider the road network in Figure \ref{fig:partition_tree_labelling}(a), partitioned into $P_A$ and $P_B$ by $V_{cut} = \{5, 12, 16\}$. $P_A$ is not distance-preserving since $\dist_{G[P_A]}(1,8)=3$  but $\dist_G(1,8)=2$. $P_B$ is distance-preserving because Equation~\ref{equ:distance} is satisfied.  
\end{example}

To characterise how the violation of this distance-preserving property by a partition relates to cut vertices, we introduce the notion of \emph{border vertex}.

\begin{definition}[Border Vertex]
Let $(P_A,P_B,V_{cut})$ be a balanced cut on a graph $G$ and $P\in\{P_A,P_B\}$. Then a vertex $v\in P$ is a \emph{border vertex} w.r.t. $V_{cut}$ iff there exists an edge $(v, w) \in E(G)$ such that $w \in V_{cut}$. 
\end{definition}
Let $\vborder(P)=\{v\in P \mid (v,w)\in E(G), w \in \vcut\}$ referring to the set of all border vertices in $P$ on a graph $G$. We have the following Lemma.
\begin{lemma}\label{lem:shortcut}
Given two vertices $x,y\in P$, the following statements are equivalent:
\begin{itemize}
\item[(1)] $\dist_{G[P]}(x,y)\neq\dist_G(x,y)$;
\item[(2)] All shortest paths between $x$ and $y$ must pass through at least one vertex in $\vcut$;
\item[(3)] All shortest paths between $x$ and $y$ must pass through at least two vertices in $\vborder(P)$.
\end{itemize}
\end{lemma}

If $P$ is not distance-preserving, according to Lemma~\ref{lem:shortcut}, it would suffice to add \emph{shortcuts} between border vertices in $P$, leading to a distance-preserving subgraph as defined below.

\begin{definition}[Shortcut-enhanced subgraph] A \emph{shortcut-enhanced subgraph} with respect to a partition $P$, denoted as $G\langle P\rangle=(P, E^*\cup E^{\Delta})$, consists of 
\begin{itemize}
    \item the set of vertices in $P$;
    \item the set of edges $E^*=\{(v_i,v_j)\in E(G) \mid v_i,v_j\in P\}$;
    \item the set of shortcuts $E^{\Delta} = \{ (b_{i}, b_{j}) \mid \{b_{i}, b_{j}\}\in \vborder(P)\}$ with $\omega(b_{i},b_{j}) = d_{G}(b_{i},b_{j})$.
\end{itemize}
\end{definition}

\begin{example}
Consider Figure \ref{fig:partition_tree_labelling}(a), in which the partition $P_A$ is not distance-preserving. A shortcut set $E^{\Delta} = \{(1, 8)\}$ with $\omega(1, 8) = 2$ is added, which leads to a shortcut-enhanced subgraph $G\langle P_A\rangle$.
\end{example}

However, not all of the shortcuts in $E^{\Delta}$ are actually needed. The following Lemma characterise redundant shortcuts:

\begin{lemma}\label{lem:redundant}
Let $b_1, b_2\in\vborder(P)$. A shortcut between $b_1$ and $b_2$ is \emph{redundant} iff one of the following conditions is satisfied:
\begin{enumerate}
\item $\dist_{G[P]}(b_1, b_2) = \dist_G(b_1, b_2)$, or
\item $\dist_G(b_1, b_2) = \dist_G(b_1, b_3) + \dist_G(b_3, b_2)$ for some border vertex $b_3\in\vborder(P)/\{b_1, b_2\}$.
\end{enumerate}
\end{lemma}

Algorithm~\ref{algo:shortcuts} shows the pseudo-code for adding shortcuts. We first compute $\dist_{cut}(b, b')$ in $P$ that is the minimal length of paths between two border vertices $b$ and $b'$ which pass through at least one cut vertex in $\vcut$ (Line 7).
The distance $d_G(b,b')$ between border vertices $b$ and $b'$ is the minimum of the lengths $d_{cut}(b,b')$ and $d_{G[P]}(b,b')$ (Line 8). After obtaining distances between all pairs of border vertices in $G[P]$ and $P$, we can check Conditions (i) and (ii) in Lemma \ref{lem:redundant} to eliminate redundant shortcuts. 

\begin{lemma}
The time complexity of Algorithm \ref{algo:shortcuts} for adding shortcuts between border vertices $B$ is
\[
O(|B|\cdot(|E|\cdot\log |V| + |B|\cdot(|B|+|V_{cut}|))).
\]
\end{lemma}
In practice the $|B|\cdot(|B|+|V_{cut}|)$ component is negligible, leaving us with an effective complexity of $O(|B|\cdot|E|\cdot\log |V|)$.

\begin{algorithm}[t]
\caption{Add Shortcuts}\label{algo:shortcuts}
\SetCommentSty{textit}
\SetKwFunction{FMain}{AddShortcuts}
\SetKwProg{Fn}{Function}{}{end}
\SetInd{0.5em}{0.5em}
\Fn{\FMain{$G$, $\vcut$, $P$}}{
    $\vborder\gets\{v\in G[P] \mid (v,w)\in E, w \in \vcut\}$ \\
    \ForEach{$b\in \vborder$}{
        $d_{G[P]} \gets$ Dijkstra's search in $G[P]$ starting from $b$ \\
        \ForEach{$b'\in \vborder\setminus\{b\}$}{
            store $d_{G[P]}(b,b')$ \\
            \tcp{distances to cut vertices already known}
            $d_{cut}(b, b')\gets\min\{ d_G(b,c) + d_G(c,b') \mid c\in\vcut \}$ \\
            $d_G(b,b')\gets min(d_{G[P]}(b,b'),\; d_{cut})$ \\
        }
    }
    \ForEach{$b_1,b_2\in \vborder$}{
        \If{$d_G(b_1,b_2) < d_{G[P]}(b_1,b_2)$}{
            redundant $\gets\false$ \\
            \ForEach{$b_3\in \vborder\setminus\{b_1,b_2\}$}{
                \If{$d_G(b_1,b_3) + d_G(b_3,b_2) = d_G(b_1,b_2)$}{
                    redundant $\gets\true$
                }
            }
            \If{not redundant}{
                add shortcut $(b_1,b_2,d_G(b_1,b_2))$
            }
        }
    }
}
\end{algorithm}

\begin{figure*}
\centering
\includegraphics[width=0.85\textwidth]{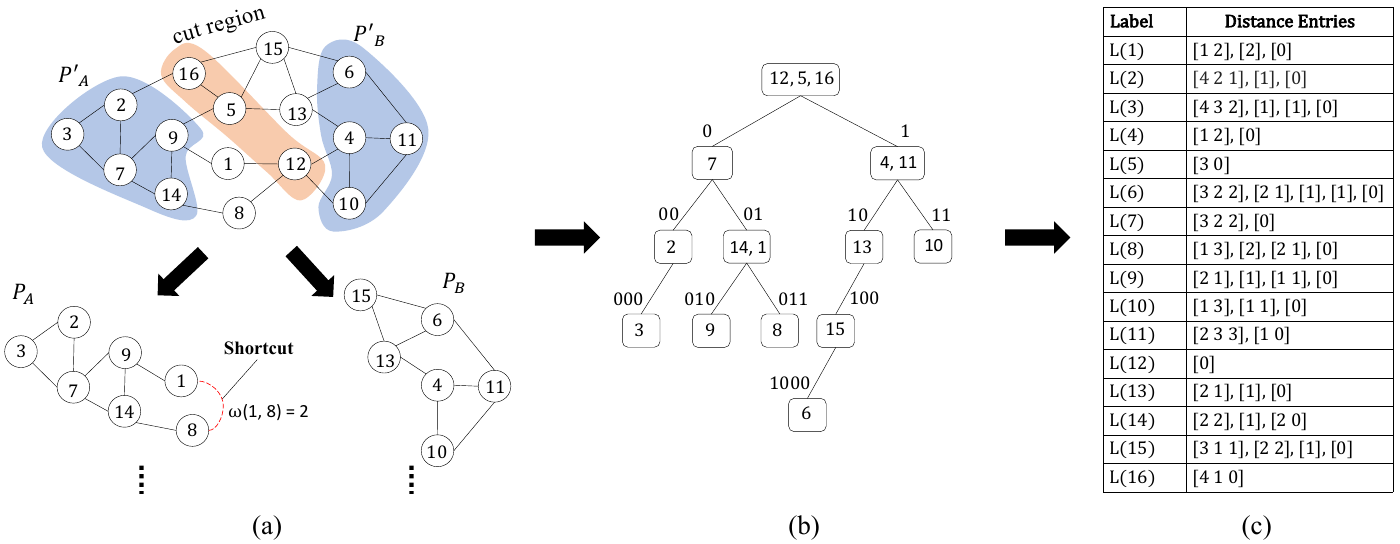}
\caption{An illustration of: (a) Hierarchical balanced cut, (b) Balanced tree hierarchy, and (c) Hierarchical cut 2-hop labelling, for the road network illustrated in Figure \ref{fig:hl_labelling}(a).}
\label{fig:partition_tree_labelling}
\end{figure*}



\subsection{Labelling Construction}
Now we present our algorithm for constructing a distance labelling, namely \emph{Hierarchical Cut 2-Hop Labelling} (HC2L). This distance labelling is said to be \emph{hierarchical} because it is constructed upon a vertex {quasi-order} defined by a balanced tree hierarchy $H_G$ over $G$.

\begin{definition}[Vertex {Quasi-Order}]
A balanced tree hierarchy $H_G$ defines a {\emph{vertex quasi-order}} $\preceq$ on $V(G)$ such that $v_i\preceq v_j$ iff $\ell(v_i)$ is an ancestor of $\ell(v_j)$ in $H_G$, {including $\ell(v_j)$ itself}.
\end{definition}

\begin{definition}[Hierarchical Cut 2-Hop Labelling] A distance labelling $L_G$ over $G$ is a \emph{hierarchical cut 2-hop labelling} (HC2L) w.r.t. $H_G$ if, 
\begin{enumerate}
\item for any label $L(v)$, $v\preceq u$ holds for any vertex $u$ in $L(v)$;
\item for any two vertices $\{s,t\}\subseteq V$
there exists $r\in LCA(s,t)$ with $(r,\delta_{sr})\in L(s),\; (r,\delta_{tr})\in L(t)$ and $\delta_{sr}+\delta_{tr}=d_G(s,t)$.
\end{enumerate}
\end{definition}

Condition (1) states that $L_G$ is \emph{hierarchical} in terms of the {vertex quasi-order} $\preceq$ defined by $H_G$.
Condition (2) stipulates that $L_G$ is a 2-hop labelling for which the distance between $s$ and $t$ can be computed from distance entries corresponding to $LCA(s,t)$ in their labels alone.



\subsubsection{Upper and Lower bounds} Assume that we are given a balanced tree hierarchy $H_G=(\mathcal{N}, \mathcal{E}, \ell)$. We start with a naive labelling algorithm to construct labels. Let $L(v)=\emptyset$ for each vertex $v\in V(G)$. Then for each tree node $N_i \in \mathcal{N}$, we conduct Dijkstra's search from every $u \in N_i$ to compute the shortest-path distance $d_G(u, v)$ and add it to $L(v)$ for all $v\in \textsc{Subtree}(N_i)$, i.e., $L(v) = L(v) \cup \{(u, d_G(u, v))\}$. Thus, for each vertex $v \in V(G)$, $L(v)$ stores $|A_v|$ distance entries, where $A_v = \{u\in V(G) | u \preceq v \}$.

\begin{lemma}
The space complexity of the naive labelling is bounded by $O(\sum_{v\in V}|A_v|)$. 
\end{lemma}

Indeed, the naive labelling approach provides the upper bound of the labelling size for hierarchical cut 2-hop labelling. It suffers from the following drawbacks:
1) it can produce very large labelling sizes for large road networks such as the whole USA and Central Europe road networks,
2) queries may perform unnecessary computations. The question arises: \emph{can we prune distance entries to {reduce labelling size and} accelerate queries?} To answer this question, we first exploit a new labelling property, called \emph{cut cover} property.

\begin{definition}[Cut Cover] Let $V_{cut}\subseteq V(G)$ be a vertex cut on $G$.
Then, for any vertex $v \in G(V)\backslash V_{cut}$ and any cut vertex $u\in V_{cut}$, $(u, \delta_{uv}) \in L(v)$ iff there is no other cut vertex $u'\in V_{cut}$ satisfying:
\begin{equation}
\dist_G(u,v) = 
\dist_G(u,u')  + \dist_G(u',v).
\end{equation}
\end{definition}

Assume that distances among two cut vertices in any vertex cut of $T_G$ are known, denoted as $\delta_{cut}$. The cut cover property ensures that the label of every non-cut vertex contains the distance information to cut vertices - either directly or indirectly through another cut vertex. Thus, the cut cover property relates to the lower bound of the labelling size for hierarchical cut labelling. 

\begin{lemma}
Let $L^*$ be a hierarchical cut 2-hop labelling satisfying the cut cover property. Then the intersection of all hierarchical cut 2-hop labellings of $H_G$ is $L^*$, i.e., $u\in L^*(v)$ iff $u\in L(v)$ for all hierarchical cut 2-hop labellings $L$ of $H_G$.
\end{lemma}


If a hierarchical cut 2-hop labelling $L$ satisfies the cut cover property, then $L$ is minimal, i.e., if removing any distance entry from $\{L(v)\}_{v\in V\backslash V_{cut}}$, then there must exist two vertices $s,t\in V(G)$ whose distance $d_G(s,t)$ cannot be computed from $\{L(v)\}_{v\in V\backslash V_{cut}}$ and $\delta_{cut}$.
Note that not all hierarchical cut 2-hop labellings can satisfy the cut cover property because it does not guarantee 2-hop labelling. Below we introduce a new labelling method, called \emph{tail pruned labelling}, which guarantees that: 1) the resultant labelling is still 2-hop, and 2) the labelling size is small and within the upper and lower bounds, thereby enabling efficient querying. 

\subsubsection{Tail Pruned Labelling}



Given a balanced cut $(P_A,P_B,V_{cut})$ on a graph $G=(V,E)$, we rank each cut vertex $v\in V_{cut}$ based on the following equation:
\begin{equation}\label{eq:cut-rank}
r(v) = |\{ u\in V \mid \exists v'\in V_{cut}\setminus\{v,u\} \text{ with } sp(v,v',u) \}|,
\end{equation}
where $sp(v,v',u)$ denotes that $v'$ lies on a shortest path between $v$ and $u$, i.e., $d_G(v,u) = d_G(v,v') + d_G(v',u)$.
Ties are broken arbitrarily to obtain a total ordering.
Intuitively, vertices in a vertex cut are assigned a rank in terms of how frequently they can be hit by other cut vertices in their shortest-paths to other vertices.
We then prune as follows.




\begin{definition}[Tail Pruning]\label{def:tail-pruning}
A cut vertex $v_i\in V_{cut}$ is tail-pruned from $L(u)$ iff
\begin{enumerate}
\item $\exists v_j\in V_{cut}$ with $r(v_j) < r(v_i)$ and $sp(v_i,v_j,u)$, and \label{cond:pruneable}
\item $\forall v_k\in V_{cut}$ with $r(v_i) < r(v_k)$, $v_k$ is tail-pruned from $L(u)$. \label{cond:tail-pruned}
\end{enumerate}
\end{definition}

Condition (1) ensures that a 2-hop labelling remains after pruning, and Condition (2) allows us to omit vertex identifiers in labels.

Algorithm \ref{algo:labelling} shows the pseudo-code for the tail pruned labelling approach. A modified version of the Dijkstra's algorithm is shown in Algorithm~\ref{algo:dist-prune} which computes the distances from a given (cut) vertex $root$ and tracks whether a shortest path from $root$ to a vertex $v\in G$ passes through another (cut) vertex in a given set $P$. In Algorithm~\ref{algo:labelling}, we compute the ranks of cut vertices as per Equation \ref{eq:cut-rank} (Lines~\ref{ln:label-rank-begin}-\ref{ln:label-rank-end}). Afterwards, Algorithm~\ref{algo:dist-prune} is invoked with $P$ containing only cut-vertices of lower ranks, allowing condition (\ref{cond:pruneable}) of Definition~\ref{def:tail-pruning} to be checked (Line~\ref{ln:label-pruneable}).
Finally, the list of distance values is tail-pruned in accordance with Definition~\ref{def:tail-pruning} (Lines ~\ref{ln:label-tail-begin}-\ref{ln:label-tail-end}).

\begin{example}
Consider $V_{cut}=\{5,12,16\}$ in Figure~\ref{fig:partition_tree_labelling}. Evaluating Equation (\ref{eq:cut-rank}) gives us the ranking $r(12) < r(5) < r(16)$.
As a result, $L(1)=\{(12,1),(5,2),(16,3)\}$ is represented as $[1,2,3]$ and can be tail-pruned to $[1,2]$.
For $L(2)=\{(12,4),(5,2),(16,1)\}$, represented as $[4,2,1]$, no tail-pruning is possible.
Although vertex $16$ lies on a shortest path between vertices $5$ and $2$, it has a greater rank than $5$, so $5\in L(2)$ does not meet Condition (\ref{cond:pruneable}) of Definition~\ref{def:tail-pruning}.
Nor does it meet Condition (\ref{cond:tail-pruned}), since $16\in L(2)$ does not meet Condition (\ref{cond:pruneable}).
\end{example}

\begin{algorithm}[t]
\caption{Dijkstra with pruneability tracking}\label{algo:dist-prune}
\SetCommentSty{textit}
\SetKwFunction{FMain}{DistAndPrune}
\SetKwProg{Fn}{Function}{}{end}
\Fn{\FMain{$G$, $root$, $P$}}{
    $dist(v\in G)\gets\infty$ \\
	add $(root,0,\false)$ to Q \\
    \While{$Q$ is not empty}{
        \tcp{$Q$ is ordered by $(d,p)$ with $\true<\false$}
		$(v,d,p)\gets pop_{\min}(Q)$\\
		\If{$dist(v)=\infty$}{
			$dist(v)\gets(d,p)$\\
			\For{$(v,w,l)\in G$}{
			    \If{$v\in P$}{
			        add $(w,d+l,\true)$ to $Q$
			    }
			    \Else{
			        add $(w,d+l,p)$ to $Q$
			    }
			}
		}
    }
    \Return $dist$
}
\end{algorithm}

\begin{algorithm}[t]
\caption{Labelling with tail pruning}\label{algo:labelling}
\SetCommentSty{textit}
\SetKwFunction{FMain}{Labelling}
\SetKwFunction{FDist}{DistAndPrune}
\SetKwProg{Fn}{Function}{}{end}
\Fn{\FMain{$G$, $\vcut$}}{
    \tcp{compute cut cover pruneability}
    \ForEach{$v\in \vcut$\label{ln:label-rank-begin}}{
        $dist\gets\FDist(G,v,\vcut\setminus \{v\})$\\
        $P_{\#}(v)\gets|\{ v\in G \mid dist(v).p=\true \}|$\\
    }
    \tcp{rank cut vertices}
    $[v_0,\ldots,v_n]\gets$ order $\vcut$ by $P_{\#}$ \label{ln:label-rank-end}\\
    \tcp{compute pruneability}
    \For{$i=0\ldots n$}{
        $dist_i\gets\FDist(G,v_i,\{v_0,\ldots,v_{i-1}\})$ \label{ln:label-pruneable}\\
    }
    \tcp{create labels with tail pruning}
    \ForEach{$v\in G$\label{ln:label-tail-begin}}{
        $k\gets\max\{ i\in 0\ldots n \mid dist_i(v).p=\false \}$\\
        $L(v)\gets [dist_0(v).d,\;\ldots,\;dist_k(v).d]$ \label{ln:label-tail-end}\\
    }
    \Return $L$
}
\end{algorithm}

We may construct a hierarchical cut 2-hop labelling by applying the pruned landmark labelling (PLL) method \cite{akiba2013fast}, which can conduct a pruned breadth-first search (BFS) from each cut vertex. 
However, PLL does not suit our data structure design for labelling because the way PLL reduces distance entries still requires a full scan on distance arrays (will discuss in the next paragraph) in the labels of vertices in a distance query. As a result, the query performance deteriorates in comparison with our tail pruned labelling method.

We implement an efficient data structure to leverage the advantages of our balanced tree hierarchy $H_G$. Since $H_G$ is a binary tree, we represent each tree node using a bitstring of length equal to its level (its distance to the root) in $H_G$.
When $\beta=\nicefrac{1}{3}$ and $G$ contains no more than $\nicefrac{2}{3}^{-58}\approx 16.3$ billion vertices, binary strings (including their 6-bit length) can be stored as 64-bit integers.
Furthermore, unlike other approaches, we only store distance values in labels, which reduces storage requirements by half (for vertex identifiers and distance values of equal size). Thus, the label of each vertex is a list of distance arrays, each corresponding to a vertex cut. This enables us to search only a 
reduced set of hub vertices from exactly one vertex cut in the labels of any query pair.

Before constructing labels, we contract the graph by repeatedly removing degree-one vertices. Most shortest-paths from a degree-one vertex to other vertices pass through its closest vertex in the contracted graph, called its \emph{root}.
For each degree-one vertex we store its distance to its root, as well as a reference to the root. This allows us to compute distances between two vertices $v$ and $w$ by using the references and then adding the distances stored.

However, this approach only works when the roots of $v$ and $w$ lie on all (shortest) paths between them, which may not be the case when $v$ and $w$ have the same root. To deal with this case efficiently, we observe that contracted vertices with the same root form a tree (with their common root designated as its root), and store for each contracted vertex its tree parent. This allows us to follow the paths from $v$ and $w$ to their lowest common ancestor $u$ in the tree using reference information alone, and compute their distance as
\[
d_G(v,w) = d_G(v,root) + d_G(w,root) - 2\cdot d_G(u, root).
\]

A similar approach is taken by PHL in \cite{akiba2014fast}. The difference is that they only prune vertices that have degree one in the original graph. This ensures that all paths between different vertices must pass through their roots, but reduces contraction effectiveness (from $\sim$30\% to $\sim$20\% for the graphs we experimented on).

\subsection{Query Processing}
Now we describe how our query function $Q$ efficiently answers distance queries on a road network $G$, given a balanced tree hierarchy $H_G$ and a hierarchical cut 2-hop labelling $L_G$. 

Given a query pair $(s,t)$ with $s,t\in V(G)$, we process the query for $(s,t)$ in two steps: (1) In the first step, we compute the lowest common ancestor $LCA(s,t)$ of the tree nodes $\ell(s)$ and $\ell(t)$ in $H_G$ (Line 2);
(2) In the second step, we compute the distance between $s$ and $t$ using the information in LCA and the label sets $L(s)$ and $L(t)$ in $L_G$ as follows. 
\begin{align}\label{eq:dist}
d_G(s,t)=&min\{\delta_{sr}+\delta_{tr}:\\\nonumber
&\delta_{sr}\in L(s),\delta_{tr}\in L(t), r\in LCA(s,t)\}.
\end{align}



As distances to cut vertices are organized by level, we only need to find the level of $LCA(s, t)$, i.e., the length of its identifying bitstring. This can be computed as the number of leading zeros of the XOR of the bitstrings of $s$ and $t$ -- operations that are natively supported by most CPUs and thus extremely fast.

\begin{example}
Consider $H_G$ and $L_G$ in Figure \ref{fig:partition_tree_labelling}(b) and \ref{fig:partition_tree_labelling}(c). The distance query $(14, 15)$ first finds $LCA(14, 15) = \{12, 5, 16\}$, the first cut in $H_G$, by comparing the bitstrings 01 and 100.
The distances between 14 and $\{12, 5, 16\}$ are $[2,2,\cancel{3}]$, which are found as the first distance array in $L(14)$, with the last distance value $3$ removed by tail-pruning. Similarly, distances between 15 and $\{12, 5, 16\}$ are found to be $[3,1,1]$. We then compute $d_G(14,6)$ as $\min(2+3, 2+1)$, ignoring the last distance value $1$ in $[3,1,1]$ as its counterpart was pruned.
\end{example}

\begin{lemma}\label{th:lca}
Given a balanced tree hierarchy $H_G$, and a query pair $(s, t)$, $LCA(s, t)$ in $H_G$ can be obtained in time $O(1)$.
\end{lemma}

\begin{lemma}
Given a distance query $(s,t)$, the distance $d_G(s,t)$ can be computed using Equation \ref{eq:dist} in
$O(V_{cut})$, where $V_{cut}$ is the largest cut in the balanced tree hierarchy $H_G$.
\end{lemma}

\subsection{Parallelization}
To speed up construction, we can use multi-threading to parallelise certain tasks. Whenever we partition a (sufficiently large) subgraph, we create a new thread in order to process two partitions in parallel. In each thread, we further parallelise computing distances for labels, shortcuts, and pruning in each thread, i.e., we perform a Dijkstra search from each cut (or border) node. These searches can easily be carried out in parallel, with workloads being practically identical. Together these tasks account for the majority of labelling construction process. We must note however that effective parallelization is limited to a handful of physical threads, as a significant part of the work (>10\%) is not or not fully parallelised -- mainly due to early cut computations.

\section{Experiments}\label{sec:experiments}

\noindent\textbf{Hardware and Platform~}All the experiments are performed on a Linux server Intel Xeon W-2175 with 2.50GHz CPU, 28 cores, and 512GB of main memory. All the algorithms were implemented in C++11 and compiled using g++ 9.4.0 with the -O3 option. 

\vspace{0.1cm}
\noindent\textbf{Datasets~}We use 10 undirected real road networks, nine of them are from the US and publicly available at the webpage of the 9th DIMACS Implementation Challenge \cite{demetrescu2009shortest} and one is from Western Europe managed by PTV AG \cite{ptvplanung}. Table \ref{table:datasets} summarizes these datasets in which the largest dataset is the whole road network in the USA. 
\begin{updated}
Further, we consider two versions for these datasets which have different units of edge weights i.e., distances and travel times. 
\end{updated}

\begin{table}[h!]
\centering
\caption{Summary of datasets.}
\label{table:datasets}\vspace{-0.3cm}
\scalebox{0.85}{
\begin{tabular}{| l l | r r r r|} 
    \hline
    Dataset & Region & $|V|$ & $|E|$ & \emph{diam.} & Memory \\
    \hline\hline
    NY & New York City & 264,346 & 733,846 & 720 & 17 MB \\
    BAY & San Francisco & 321,270  & 800,172 & 721 & 18 MB \\
    COL & Colorado & 435,666 & 1,057,066 & 1,245 & 24 MB \\
    FLA & Florida & 1,070,376 & 2,712,798 & 2,058 & 62 MB \\
    CAL & California & 1,890,815 & 4,657,742 & 2,315 & 107 MB \\
    E & Eastern USA & 3,598,623 & 8,778,114 & 4,461 & 201 MB \\
    W & Western USA & 6,262,104 & 15,248,146 & 4,420 & 349 MB \\
    CTR & Central USA & 14,081,816 & 34,292,496 & 5,533 & 785 MB \\
    USA & United States & 23,947,347 & 58,333,344 & 8,440 & 1.30 GB \\ 
    EUR & Western Europe & 18,010,173 & 42,560,279 & 5,175 &  974 MB \\
    \hline
\end{tabular}}
\end{table}

\begin{table*}[ht]
 \centering
 \caption{Comparison of query times, labelling sizes and construction times with distances as edge weights between our method, i.e., HC2L, and the baseline methods. HC2L$^p$ represents HC2L with parallelism.}
 \label{table:performance_dist}\vspace{-0.3cm}
 \begin{tabular}{| l || r r r r | r r r r | r r r r r|}  \hline
	\multirow{2}{*}{Dataset}&\multicolumn{4}{c|}{Query Time [$\mu$s]}&\multicolumn{4}{c|}{Labelling Size}&\multicolumn{5}{c|}{Construction Time [s]} \\\cline{2-14}
    & HC2L & H2H & PHL & HL & HC2L & H2H & PHL & HL & HC2L & HC2L$^p$ & H2H & PHL & HL \\

   \hline\hline
    NY & 0.225 & 0.432 & 0.983 & 0.765 & 144 MB & 341 MB & 320 MB & 233 MB & 15 & 6 & 16 & 34 & 32 \\
    BAY & 0.220 & 0.563 & 0.707 & 0.665 & 113 MB & 339 MB & 235 MB & 219 MB & 12 & 4 & 12 & 18 & 27 \\
    COL & 0.351 & 0.750 & 0.909 & 0.720 & 236 MB & 217 MB & 403 MB & 341 MB & 27 & 12 & 21 & 38 & 45 \\
    FLA & 0.371 & 0.754 & 0.965 & 0.827 & 487 MB & 1.25 GB & 1.14 GB & 907 MB & 68 & 23 & 46 & 121 & 137 \\
    CAL & 0.442 & 1.125 & 1.106 & 0.958 & 1.24 GB & 3.87 GB & 2.58 GB & 1.78 GB & 215 & 57 & 146 & 327 & 318 \\
    E & 0.555 & 1.241 & 1.671 & 1.218 & 3.37 GB & 9.81 GB & 8.44 GB & 4.74 GB & 654 & 163 & 409 & 1,578 & 1,149 \\
    W & 0.583 & 1.382 & 1.661 & 1.163 & 5.71 GB & 18.3 GB & 13.5 GB & 7.50 GB & 1,197 & 261 & 702 & 2,314 & 1,654 \\
    CTR & 0.760 & 1.630 & 2.503 & 1.613 & 24.4 GB & 73.9 GB & 55.9 GB & 25.5 GB & 6,203 & 1,658 & 4,029 & 15,882 & 7,591 \\
    USA & 0.737 & 1.940 & 2.389 & 1.663 & 45.1 GB & 155 GB & 95.6 GB & 44.7 GB & 11,203 & 1,977 & 7,737 & 26,515 & 13,157 \\ 
    EUR & 0.922 & 2.414 & 2.239 & 1.673 & 44.1 GB & 160 GB & 70.9 GB & 34.1 GB & 12,242 & 3,083 & 9,194 & 20,466 & 8,728 \\\hline
 \end{tabular}
\end{table*}

\begin{table}[t]
 \centering
 \caption{Comparing LCA storage requirements and average hub size by HC2L and the baseline algorithms for all the datasets with distances as edge weights.}
 \label{table:lca_hopsize}\vspace{-0.3cm}
 \setlength\tabcolsep{4pt}
 \begin{tabular}{| l || r r | r r r r r |}  \hline
	\multirow{2}{*}{Dataset}&\multicolumn{2}{c|}{LCA Storage}&\multicolumn{5}{c|}{Average Hub Size (AHS)} \\\cline{2-8}
    & HC2L & H2H & HC2L & P2H & H2H & PHL & HL \\

   \hline\hline
    NY & 2.02 MB & 40.3 MB & 12 & 19 & 34 & 50 & 137 \\
    BAY & 2.45 MB & 49.0 MB & 9 & -- & 51 & 22 & 105 \\
    COL & 3.32 MB & 66.5 MB & 20 & 24 & 66 & 39 & 118 \\
    FLA & 8.17 MB & 180 MB & 13 & 21 & 50 & 37 & 132 \\
    CAL & 14.4 MB & 317 MB & 23 & 36 & 109 & 39 & 154 \\
    E & 27.5 MB & 632 MB & 36 & 42 & 164 & 98 & 213 \\
    W & 47.8 MB & 1.12 GB & 37 & 67 & 171 & 99 & 192 \\
    CTR & 108 MB & 2.62 GB & 66 & 101 & 225 & 185 & 132 \\
    USA & 183 MB & 4.64 GB & 62 & 125 & 331 & 160 & 138 \\ 
    EUR & 137 MB & 3.49 GB & 142 & -- & 749 & 154 & 310 \\\hline
 \end{tabular}
\end{table} 

\vspace{0.1cm}
\noindent\textbf{Baseline Methods~}
We compared our proposed algorithm HC2L with four state-of-the-art algorithms for shortest-path distance queries in road networks as follows, 1) Hub Labelling (HL) \cite{abraham2012hierarchical}, 2) Pruned Highway Labelling (PHL) \cite{akiba2014fast}, 3) Hierarchical 2-Hop Labelling (H2H) \cite{ouyang2018hierarchy}, and 4) Projected Vertex Separator Based 2-Hop Labelling (P2H) \cite{chen2021p2h}.

The code for H2H, PHL and HL was publicly available and implemented in C++. Unfortunately, we could not obtain the implementation of P2H; we reported only the part of experimental results presented in \cite{chen2021p2h}. They implemented P2H in C++ and their experiments were conducted on a machine with Quad Intel(R) Xeon(R) Platinum 8160 24-core @ 2.10GHz CPU and 768 GB RAM, running CentOS Linux 7. We use the same parameter settings as suggested by the authors of these methods, unless otherwise stated. We select balance partition threshold $\beta = 0.2$ for HC2L. Furthermore, we set the number of threads equal to the available 28 cores.


\vspace{0.1cm}
\noindent\textbf{Benchmark Generation~}To evaluate the query performance, we randomly sampled 1,000,000 pairs of vertices from all pairs of vertices in each road network, i.e., $V\times V$.
We also evaluate the query performance of the algorithms by varying the distance between the source and target vertices in a query similar to \cite{ouyang2018hierarchy,Pohl1969BidirectionalAH}.
Specifically, for each road network, we generate 10 sets of queries $Q_1, Q_2,\dots, Q_{10}$ as follows: we set $l_{min}$ to be 1000 meters, and set $l_{max}$ to be the maximum
distance of any pair of vertices in the map. Let $x = (\frac{l_{max}}{l_{min}})^{1/10}$. For each $1 \leq i \leq 10$, we generate 10,000 queries to form each set $Q_i$, in which the distance of the source and target vertices for each query
falls in the range $(l_{min} \cdot x^{i-1},\; l_{min} \cdot x^{i}]$. For each algorithm, we report the average query processing time.

\subsection{Performance Evaluation}
We compare the performance of our proposed algorithm with the baseline methods in terms of the query time, labelling size, and construction time. The experimental results are presented in Table~\ref{table:performance_dist}, Table~\ref{table:performance_time}, and Figure \ref{fig:varying_distanceQuery}.

\begin{table*}[ht]
\begin{updated}
 \centering
 \caption{Comparison of query times, labelling sizes and construction times with travelling times as edge weights between our method and the baseline methods.}
 \label{table:performance_time}\vspace{-0.3cm}
 \begin{tabular}{| l || r r r r | r r r r | r r r r r|}  \hline
	\multirow{2}{*}{Dataset}&\multicolumn{4}{c|}{Query Time [$\mu$s]}&\multicolumn{4}{c|}{Labelling Size}&\multicolumn{5}{c|}{Construction Time [s]} \\\cline{2-14}
    & HC2L & H2H & PHL & HL & HC2L & H2H & PHL & HL & HC2L & HC2L$^p$ & H2H & PHL & HL \\

   \hline\hline
    NY & 0.210 & 0.531 & 0.574 & 0.593 & 133 MB & 272 MB & 137 MB & 152 MB & 16 & 6 & 14 & 11 & 19 \\
    BAY & 0.213 & 0.537 & 0.475 & 0.549 & 106 MB & 258 MB & 103 MB & 154 MB & 13 & 4 & 11 & 7 & 18 \\
    COL & 0.327 & 0.815 & 0.554 & 0.560 & 189 MB & 523 MB & 166 MB & 221 MB & 26 & 14 & 21 & 12 & 27 \\
    FLA & 0.361 & 0.776 & 0.625 & 0.577 & 466 MB & 1.07 GB & 478 MB & 606 MB & 74 & 23 & 46 & 40 & 84 \\
    CAL & 0.427 & 1.033 & 0.638 & 0.697 & 1.16 GB & 3.11 GB & 905 MB & 1.07 GB & 231 & 51 & 130 & 78 & 173 \\
    E & 0.479 & 1.203 & 0.820 & 0.819 & 3.01 GB & 8.47 GB & 2.41 GB & 2.43 GB & 706 & 176 & 360 & 256 & 503 \\
    W & 0.551 & 1.281 & 0.837 & 0.821 & 5.29 GB & 16.7 GB & 3.85 GB & 4.10 GB & 1,277 & 315 & 693 & 401 & 770 \\
    CTR & 0.701 & 2.054 & 1.037 & 0.977 & 23.2 GB & 84.4 GB & 11.3 GB & 11.4 GB & 7,296 & 1,728 & 4,127 & 1,513 & 2,954 \\
    USA & 0.723 & 1.864 & 1.141 & 1.018 & 38.4 GB & 135 GB & 22.4 GB & 20.3 GB & 11,249 & 2,039 & 6,988 & 3,124 & 4,752 \\ 
    EUR & 0.872 & 2.720 & 1.297 & 1.106 & 38.3 GB & 169 GB & 20.9 GB & 17.9 GB & 13,292 & 3,116 & 10,649 & 3,488 & 4,502 \\\hline
 \end{tabular}
\end{updated}
\end{table*}
\begin{table}[t]
 \centering
 \caption{Comparing Tree Height and Max Cut Size on all the datasets with distances as edge weights.}
 \label{table:height_width}\vspace{-0.3cm}
 \begin{tabular}{| l || c c | c c |}  \hline
	\multirow{2}{*}{Dataset}&\multicolumn{2}{c|}{Tree Height}&\multicolumn{2}{c|}{Max Cut Size/Width} \\\cline{2-5}
    & HC2L & H2H/P2H & HC2L & H2H/P2H \\

   \hline\hline
    NY & 24 & 399 & 40 & 171 \\
    BAY & 25 & 297 & 41 & 120 \\
    COL & 27 & 376 & 42 & 192 \\
    FLA & 29 & 345 & 34 & 140 \\
    CAL & 30 & 603 & 56 & 266 \\
    E & 30 & 869 & 75 & 328 \\
    W & 32 & 895 & 89 & 342 \\
    CTR & 35 & 1,771 & 153 & 667 \\
    USA & 35 & 2,135 & 178 & 855 \\ 
    EUR & 36 & 2,892 & 280 & 1,181 \\\hline
 \end{tabular}
\end{table} 

\subsubsection{Querying Time}
In Table~\ref{table:performance_dist}, we report for each method and dataset the average query time over 1 million random distance queries, \begin{updated}where edge weights are distances\end{updated}. We confirm that HC2L is the fastest on all the datasets.
In most cases, we are 1.5-2.5 times faster than HL, 2-3 times faster than H2H, and 2-4 times faster than PHL. The is because we process a significantly smaller number of distance entries in the labels of a query pair when computing their shortest-path distance.
\begin{updated}
Table~\ref{table:performance_time} reports the results when we consider travel times, rather than distances, as edge weights for each dataset. We notice that the average query times for almost all the methods become faster compared to their corresponding results in Table~\ref{table:performance_time}. The reason for this speedup lies in the reduced labelling sizes which will be further discussed in Section~\ref{subsubsec:labelling size}. Evidently, HC2L is still the fastest on all the datasets.
\end{updated}

Table \ref{table:lca_hopsize} shows the average number of hubs for which the sum of distances is computed when evaluating Eq.~(\ref{equ:2-hop}) or Eq.~(\ref{eq:dist}). Compared to H2H, HC2L produces minimal cuts which are very small in practice as shown in Table \ref{table:height_width}, and thus the search space of HC2L is significantly reduced for distance queries. The query time of PHL and HL depends on the label sizes of a query pair. We can see from Table \ref{table:lca_hopsize} that the average number of hubs of PHL and HL are much bigger than HC2L, which make them significantly underperform HC2L.
Furthermore, PHL is generally slower than HL as it employs highways instead of vertices as hubs, making distance computation more complex than simply adding two distance values.

\vspace{0.1cm}
\noindent\textbf{Varying Distance Querying~}
\begin{updated}
\end{updated}
In Figure \ref{fig:varying_distanceQuery}, we report results for distance query sets containing query pairs with varying distances to test the performance of all the algorithms. We can clearly see that HC2L significantly outperforms all the baseline methods in every query set. Particularly, our method HC2L shows good performance for both short and long distance query sets.
Regardless of distance, only cut vertices of the LCA of the query vertices need to be considered as hubs. Vertex cuts tend to be smaller at lower levels of the hierarchy, thus making local queries generally faster. Exceptions to this exist though -- the NY dataset has a top-level cut of size 5, making distant queries very fast as well. Additionally, queries involving top-level cuts enjoy better caching, due to the memory layout of labels.
H2H and HL show similar behaviour. In contrast, PHL has rather poor query performance for local queries -- this happens because it optimizes its index for high-speed highways, which are less likely to intersect with shortest paths between local vertex pairs. 
\begin{updated}
Optimizing indexes for real-world workloads is an interesting problem that often requires specific algorithmic designs as discussed in \cite{DBLP:conf/icde/ZhengWGMHZJ22}.
\end{updated}

\subsubsection{Labelling Size}\label{subsubsec:labelling size}
Table~\ref{table:performance_dist} shows that the labelling sizes produced by our proposed method HC2L is significantly smaller than the baseline methods when indexing for distance queries.
In most cases, the labelling size of HC2L is about 2-4 times smaller than H2H, 2-3 times smaller than PHL and 1-2 times smaller than HL.
The only exception is the EUR dataset, for which HL produces a smaller labelling.
Despite this, HC2L still beats HL in terms of query time on EUR. This is because HC2L only uses a fraction of labels during query answering, although the gap does become smaller.
We have also analysed the extra overhead of labelling sizes by H2H and HC2L to find LCA in constant time in Table~\ref{table:lca_hopsize}.
It shows that the overheads incurred by H2H are about 20 times greater -- this happens because the tree hierarchies of H2H are neither binary nor balanced, thus requiring a different indexing approach.
We found that without tail pruning index sizes grow by 10-15\%, but construction time is reduced by around 20\%.
\begin{updated}
Table~\ref{table:performance_time} shows labelling sizes for datasets using travel times as edge weights. Here HC2L produces slightly smaller labelling sizes than the corresponding labelling sizes in Table~\ref{table:performance_dist}, H2H is roughly the same, but for PHL and HL labelling sizes are reduced significantly when changing edge weights from distances to travelling times. The reason for this is that PHL and HL can exploit better orderings using travel times as edge weights, which leads to better pruning.
Further, we notice that HL pruning is more effective than the tail pruning approach employed by HC2L. This is because HL pruning affects the whole vertex ordering, while the tail pruning by H2CL only affects the vertex ordering within each cut; thus, one can expect HL to perform better in the cases where a large percentage of labels can be pruned.
\end{updated}

\subsubsection{Construction Time}
In Table~\ref{table:performance_dist} we compare the construction time of our method HC2L with the baseline methods when indexing for distance.
We observe that the single-threaded implementation of HC2L is slower than H2H, faster than PHL and comparable to HL, but the parallel variant of our method, denoted by HC2L$^p$, significantly outperforms the baseline methods on all the datasets.
The construction time of our algorithms includes both the time for obtaining a balanced tree hierarchy and the time for constructing a hierarchical cut labelling.
\begin{updated}
When considering travel times as edge weights, as shown in Table~\ref{table:performance_time}, PHL and HL perform significantly faster in construction, compared to considering distances as edge weights in Table~\ref{table:performance_dist}; nonetheless, these methods are still slower than our parallel implementation.
This increase in construction time is consistent with the decrease in labelling size due to better pruning, as discussed earlier.
\end{updated}

\begin{figure*}[ht!]
\centering
\includegraphics[width=0.92\textwidth]{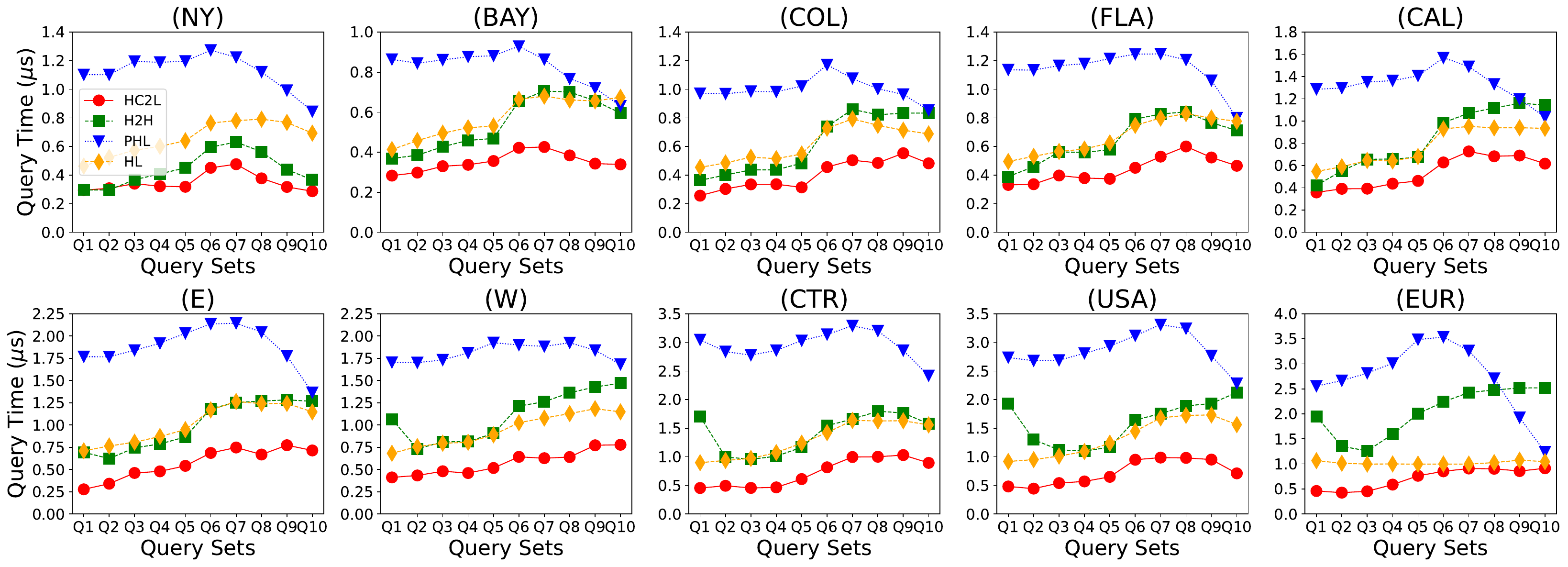}\vspace{-0.3cm}
\caption{Distance query performance under varying distances for all the datasets with distances as edge weights.}
\label{fig:varying_distanceQuery}\vspace{-0.15cm}
\end{figure*}
\begin{figure*}[ht]
\includegraphics[width=0.92\textwidth]{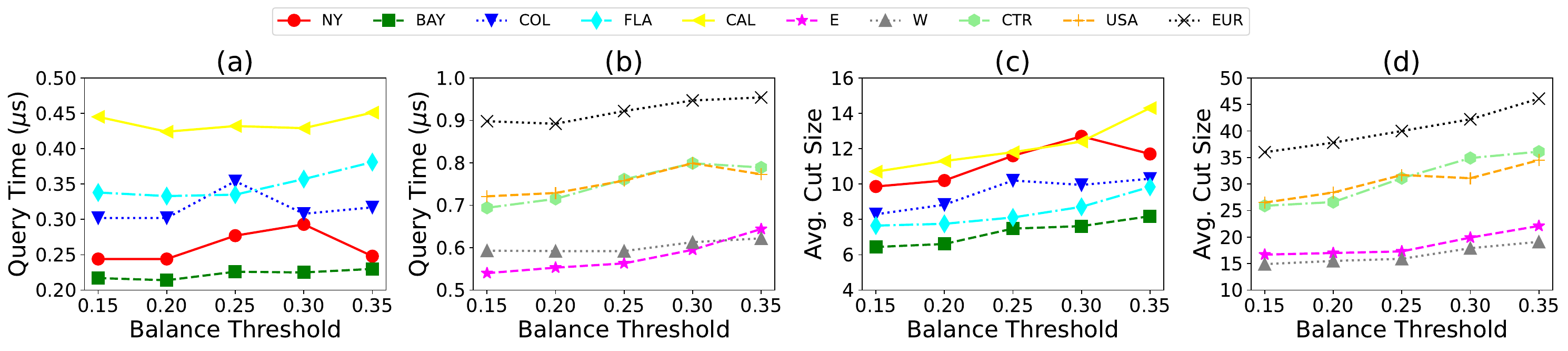}\vspace{-0.3cm}
\caption{Query performance with respect to cut size under varying balance partition thresholds for all the datasets with distances as edge weights.}
\label{fig:query_time_wrt_balance}
\end{figure*}

\subsection{Analysis of the Proposed Approaches}
We evaluate how the balance threshold $\beta$, cut size and tree height affect the performance of our method.
Figure \ref{fig:query_time_wrt_balance} shows the average query times and cut sizes under varying balance thresholds $\beta$ on all the datasets. We can observe that whenever a cut size increases/decreases under a particular threshold, a similar trend reflects in query time of that threshold. The average query times decrease or remain the same on the majority of datasets when we increase $\beta$ from $0.15$ to $0.20$ and for thresholds over $0.20$ seem increasing. This also aligns well with the cut sizes, and using the larger balance threshold $0.20$ leads to more balanced trees. 

In Table \ref{table:height_width}, we also report the height and maximum cut size/width of our method HC2L with a balance threshold $0.20$ against the baseline methods H2H and P2H. It is clear that HC2L produces much smaller cuts and has significantly smaller heights for all datasets compared to H2H and P2H. This is because our method is based on novel techniques to find balanced partitioning and small cuts which result in a balanced binary tree with a very small height.
In contrast, H2H and P2H use heuristic techniques to decompose a road network into a tree which may produce arbitrary height and width of a very large number in practice as shown in Table \ref{table:height_width}.

\subsection{Extension to Directed Graphs}

Our approach can be extended to directed graphs by storing distances from both directions in each label.
We can compute them by performing searches on both directions  when adding shortcuts as well as conducting label construction.
To compute vertex cuts, we can treat all edges as undirected to ensure that they separate paths in both directions.
However, road networks tend to be \emph{almost} undirected (with some famous exceptions such as Stockholm), in which case the two distances stored within each label will frequently be identical -- this can be exploited for optimizations.

\begin{updated}
\subsection{Remarks}
Below we comment on some possible extensions and open problems related to our method proposed in this work:
\begin{itemize}[leftmargin=*]

\item Although our parallel implementation reduces construction time to less than one hour for all graphs, 
further increasing the number of cores has a limited impact on performance. This is because most cores would remain poorly utilized. Particularly, better utilization requires the vertex cut computation to be parallelizable, which seems tricky to do and is outside the scope of this paper. 

\item In dynamic settings, e.g., due to roads being temporarily closed or suffering from slower travel speeds, we need to update our labellings, preferably without recomputing them from scratch. 
Our balanced tree hierarchy construction does not depend on edge weights, except for shortcuts. This should enable us to preserve a balanced tree hierarchy (with some adjustments for shortcuts) and limit updates only to distance values. This is in contrast to existing approaches such as HL or PHL which rely on edge weights for node or highway ordering.
\item Existing labelling-based methods for distance queries on road networks still require labellings of sizes larger than the original graph. Thus, how to reduce labelling sizes needed for distance queries on road networks without compromising query performance still remains an open problem. 
\end{itemize}
\end{updated}
\section{Conclusion}
In this paper, we analyse the drawbacks of the current state-of-the-art labelling-based solutions in order to exploit hierarchical structures of road networks for efficiently answering distance queries.
We  propose a novel solution called hierarchical cut labelling (HC2L) to overcome the drawbacks of existing solutions. Our proposed solution uses a novel \emph{balanced tree hierarchy} to find a partial vertex order which helps in significantly reducing labelling size and selecting a small subset of labels for a distance query pair.
Accordingly, query processing is accelerated, regardless of distance between the node pairs queried.
As demonstrated experimentally, index size and construction time are competitive as well.
For future work, we plan to investigate dynamic updates to our index structure, and efficient algorithms for the balanced minimal S-T cut problem.

\clearpage

\bibliographystyle{ACM-Reference-Format}
\balance
\bibliography{references}

\end{document}